\documentclass[twocolumn,10pt]{IEEEtran}
\usepackage{braket}

\input{def.tex}
\usepackage{dsfont}
\usepackage{minibox}
\DeclareSymbolFont{matha}{OML}{txmi}{m}{it}
\DeclareMathSymbol{\varv}{\mathord}{matha}{118}
\usepackage{eurosym}
\usepackage{hhline,physics}

\usepackage{multicol}
\usepackage{algorithm}
\usepackage{graphicx}
\usepackage{textcomp}
\usepackage{caption,pifont}
\usepackage[multiple]{footmisc}

\usepackage{algorithmic}

\makeatletter
\renewcommand{\baselinestretch}{0.88}

\IEEEoverridecommandlockouts
\begin{document}
	\title{Near-Field Beamforming for Stacked Intelligent Metasurfaces-assisted MIMO Networks } 
	\author{Anastasios Papazafeiropoulos, Pandelis Kourtessis, Symeon Chatzinotas, Dimitra I. Kaklamani, 			Iakovos S. Venieris \thanks{A. Papazafeiropoulos is with the Communications and Intelligent Systems Research Group, University of Hertfordshire, Hatfield AL10 9AB, U. K., and with SnT at the University of Luxembourg, Luxembourg.  P. Kourtessis is with the Communications and Intelligent Systems Research Group, University of Hertfordshire, Hatfield AL10 9AB, U. K.  S. Chatzinotas is with the SnT at the University of Luxembourg, Luxembourg. Dimitra I. Kaklamani is with the Microwave and Fiber Optics Laboratory, and Iakovos S. Venieris is  with the Intelligent Communications and Broadband Networks Laboratory, School of Electrical and Computer Engineering, National Technical University of Athens, Zografou, 15780 Athens,	Greece.	
							A. Papazafeiropoulos was supported  by the University of Hertfordshire's 5-year Vice Chancellor's Research Fellowship. S. Chatzinotas   was supported by the National Research Fund, Luxembourg, under the project RISOTTI.
			Corresponding author's email: tapapazaf@gmail.com.}}
	\maketitle\vspace{-1.7cm}
	\begin{abstract}	
Stacked intelligent metasurfaces (SIMs) \textcolor{black}{have} recently gained significant interest since they enable precoding in the wave domain that comes with increased processing capability and reduced energy consumption. The study of SIMs and high frequency propagation make the study of the performance in the near field of crucial importance. Hence, in this work, we focus on SIM-assisted multiuser multiple-input multiple-output (MIMO) systems operating in the near field region. To this end, we formulate the weighted sum rate maximisation problem in terms of the transmit power and the phase shifts of the SIM. By applying a block coordinate descent (BCD)-relied algorithm, numerical results show the enhanced performance of the SIM in the near field with respect to the far field.
	\end{abstract}
	\section{Introduction}
Reconfigurable intelligent surface (RIS) has recently emerged as a fundamental technology that increases network capacity while accounting for energy sustainability  \cite{DiRenzo2020}. Generally, an RIS consists of an artificial surface with a large number of nearly passive elements that can shape the propagation environment \cite{Huang2019,Papazafeiropoulos2023}.

Most existing works on RIS have relied on the assumption of a single-layer surface, which limits the degrees of freedom concerning the adjustment of the beam patterns  \cite{Huang2019,Papazafeiropoulos2023, Guo2020a}. Moreover, it has been shown that conventional RISs do not have the capability of inter-user interference suppression  \cite{Guo2020a}.

These observations motivated the proposition of stacked intelligent metasurface (SIM), which includes a stack of an array of intelligent surfaces similar to the structure of artificial neural networks  \cite{An2023b}. Note that a SIM is not a mathematical abstraction. Among its remarkable properties, we draw attention to its processing capability, where the forward propagation eventuates at the speed of light. Specifically, a SIM-based transceiver of point-to-point multiple-input multiple-output (MIMO) communication systems has been proposed in  \cite{An2023b}, where the combining and the precoding are implemented as the electromagnetic (EM) waves propagate through the corresponding SIMs. In \cite{An2023c} and \cite{Papazafeiropoulos2024c}, a SIM is integrated at the base station (BS) to enable beamforming in the wave domain based on instantaneous and statistical CSI, respectively. In \cite{Papazafeiropoulos2024a} and \cite{Papazafeiropoulos2024}, contrary to  \cite{An2023b} and \cite{An2023c}, we have proposed more general hybrid wave-digital architectures and more efficient algorithms that enable the simultaneous optimization of all parameters.

In parallel, as we move on to higher frequencies and to larger services, the region of a near field may include distances of several hundred metres \cite{Bjoernson2020a}. Notably, in the near-field region, EM waves exhibit distinct propagation characteristics compared to the far field, e.g., from planar-wave
propagation we result in spherical-wave propagation \cite{Mu2024}. The spherical wave propagation in the near-field
introduces a new distance dimension, which facilitates interference mitigation and increases the performance. Although several works have studied the performance of RIS in the near-field region \cite{Bjoernson2020a,Mu2024}, no work has studied studied the performance of SIM in this region, which motivates this work.

In this work, we consider a SIM-assisted multiuser MIMO operating the near field.\footnote{\textcolor{black}{As mentioned before, contrary to an RIS, a SIM has an enhanced capability of adjusting the beam patterns,  can suppress the multi-user interference, the  forward propagation eventuates at the speed of light, and the precoding takes place as the EM waves propagate through 		it. The latter two mean that there is no need for digital beamforming with the accompanied radio frequency (RF) chains, which significantly reduces the hardware cost and energy consumption, while the  precoding delay is reduced since the processing performed in the		wave domain. } } In other words, contrary to previous works \cite{An2023b,An2023c,Papazafeiropoulos2024a,Papazafeiropoulos2024,Papazafeiropoulos2024c}, we focus on the performance in the near field and we assume multi-antenna users. We develop the weighted sum rate maximisation problem to  optimize both the transmit power and the SIM by applying a block coordinate descent (BCD)-based algorithm. Among the results, we observe that the use of a SIM in the near field results in increased weighted sum rate compared to far-field beamforming.

\textcolor{black}{	\textit{Notation}: Matrices  and  vectors are represented by boldface upper  and lower case symbols, respectively. The notations $(\cdot)^\T$, $(\cdot)^\H$, and $\mathrm{tr}\!\left( {\cdot} \right)$ denote the transpose, Hermitian transpose, and trace operators, respectively. Also, the symbol  $\EE\left[\cdot\right]$ denotes  the expectation operator. The notation  $\diag\left(\bA\right) $ represents a vector with elements equal to the  diagonal elements of $ \bA $. The notation  $\bb \sim \cC\cN{(\b0,\mathbf{\Sigma})}$ represents a circularly symmetric complex Gaussian vector with zero mean and a  covariance matrix $\mathbf{\Sigma}$.}

	\section{System and Channel Models}\label{System}
In this section, we present the system and channel models for the SIM-assisted near-field multi-user MIMO communication system.
\subsection{System Model}
\begin{figure}
	\begin{center}
		\includegraphics[width=0.8\linewidth]{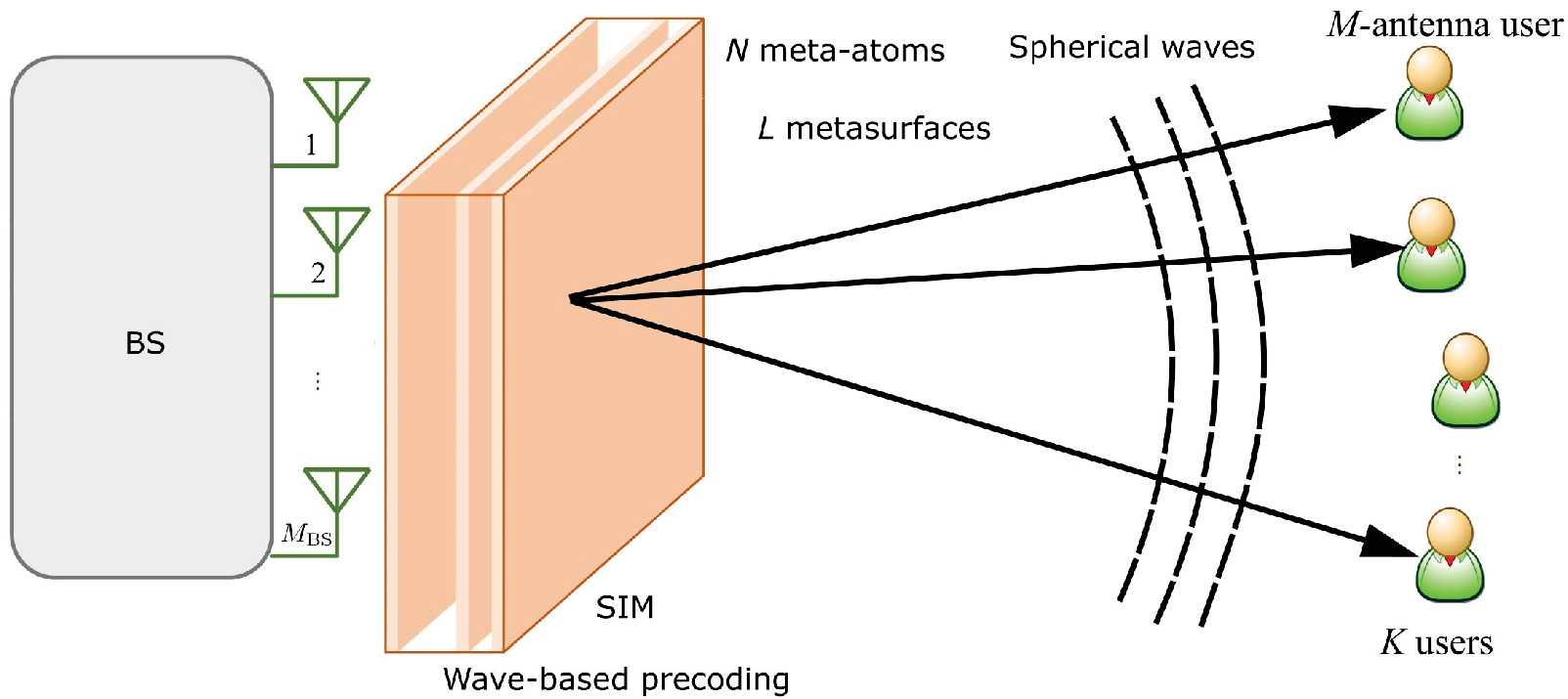}
		\caption{A SIM-aided MIMO system in the near field. }
		\label{Fig1}
	\end{center}
\end{figure}
	We consider the downlink of a multiuser MIMO wireless system with a BS  equipped with  a  uniform linear array (ULA) of $ M_{\mathrm{BS}} $ antennas that serve $ K $ multi-antenna users, each having $M$ antennas. A SIM is incorporated into the BS to provide precoding in the EM wave domain, as shown in Fig. \ref{Fig1}. The  SIM  is implemented by an array of $ L $  metasurfaces that each one of them is made by a uniform planar array (UPA) of $ N=N_{y}\times N_{z} $ meta-atoms with $N_{y}$ and $N_{z}$ denoting the number of meta-atoms in the directions of $y$ and $z$ axes. We denote the sets  $ \mathcal{N}=\{1,\ldots,N\} $, $ \mathcal{L}=\{1,\ldots,L\} $, $\mathcal{M}=\{1,\ldots,M\} $, and $\mathcal{K}=\{1,\ldots,K\} $. All metasurfaces consist of an identical number of  meta-atoms. Also,  the SIM   is  connected to a smart controller to optimize the   phase shifts of the EM waves  propagated through each of its meta-atoms.  We denote $ \bPhi^{l}=\diag(\bphi_{l})\in \mathbb{C}^{N \times N} $, where $ \bphi_{l} =[\phi^{l}_{1}, \dots, \phi^{l}_{N}]^{\T}\in \mathbb{C}^{N \times 1}$ is the  phases-shift matrix of the SIM. In particular, we have  $ \phi_{n}^{l} =e^{j \theta_{n}^{l}}$, where  $ \theta_{n}^{l}\in [0,2\pi), n \in \mathcal{N}, l \in \mathcal{L} $,  the phase shift by the $ n $-th meta-atom on the $ l $-th transmit metasurface layer. We assume that the phase shifts are continuously-adjustable. We have assumed that the modulus of the elements of the coefficient matrix  equals $ 1 $ to assess the  maximum  achievable rate \cite{Wu2019}. 
 
 \subsection{Channel Model}
  The  coefficient of transmission from the $ \tilde{n} $th meta-atom on the $ (l-1) $\textcolor{black}{th}  layer to the $ n $th meta-atom on the $ l $th  layer of the SIM is provided by the Rayleigh-Sommerfeld diffraction theory \cite{Lin2018} as  
 \begin{align}
 	w_{n,\tilde{n}}^{l}=\frac{A_{t}\cos x_{n,\tilde{n}}^{l}}{r_{n,\tilde{n}}^{l}}\left(\frac{1}{2\pi r^{l}_{n,\tilde{n}}}-j\frac{1}{\lambda}\right)e^{j 2 \pi r_{n,\tilde{n}}^{l}/\lambda}, l \in \mathcal{L},\label{deviationTransmitter}
 \end{align}
 where $ A_{t} $ denotes the area of each meta-atom,   $ r_{n,\tilde{n}}^{l} $, is the  transmission distance,and  $ x_{n,\tilde{n}}^{l} $ is the angle between the propagation direction and the normal direction of the $ (l-1) $th  metasurface layer,. Overall,  the SIM can be modeled as
 \begin{align}
 	\bG=\bPhi^{L}\bW^{L} \cdots \bPhi^{2}	\bW^{2} \bPhi^{1} \in \mathbb{C}^{N \times N },\label{TransmitterSIM}
 \end{align}
 where $ \bW^{l}\in \mathbb{C}^{N \times N}, l \in \mathcal{L}/\{1\} $ is the  coefficient matrix of transmission between the $ (l-1)st $  layer and the $ l $th  layer. \textcolor{black}{In the case of $ \bW^{1}_{k} \in \mathbb{C}^{N \times M_{\mathrm{BS}}} $, it is the transmission coefficient matrix between the $ M_{\mathrm{BS}} $ BS antennas  and the first  layer of the SIM.}
 
 Regarding the description of the near-field LoS channel \textcolor{black}{and without any loss of generality}, 
we assume that the SIM is implemented in the $YZ$-plane while all the users are located in the $XY$-plane. The reference element of the SIM and the reference antenna of user $k$ are given by $(0,y_{s}, z_{s})$ and $(x_{k}, y_{k}, 0)$, respectively. The  \textcolor{black}{coordinates} of the element $n$  of the outer surface of the SIM are given by
 \begin{align}
 	\bs_{n}=[0, c_{y}(n)d_{S}+y_{s}, c_{z}(n)d_{S}+z_{f}]^{T},
 \end{align}
 where $c_{y}(n)=\mod(n-1, N_{y})$ and $c_{z}(n)=\lfloor n-1/N_{y}\rfloor$ and  $d_{S} = \lambda/2$ is the SIM element spacing. In the case of the ULAs of users, we assume that they are parallel to the $y$-axis. Hence, the  coordinates of antenna $m$ of user $k$ are given by
 \begin{align}
 	\bu_{m}^{k}=[x_{k}, (m-1)d_{m}+y_{z},0]^{T},
 \end{align}
 where $d_{m} = \lambda/2$ is the user antenna spacing. 
 
 The near-field LoS channel between the $n$-th STAR element and the $m$-th antenna of user $k$, where the  EM signals  undergo free space path-loss while propagating in spherical wavefronts, 
 is given by \cite{Zhang2022}
 \begin{align}
 [\bH_{k}]_{mn}=\al_{mn}^{k}\exp(-j 2 \pi r_{mn}^{k}/\lambda),
 \end{align}
 where $ r_{mn}^{k}=\|\bu_{m}^{k}-	\bs_{n}\|_{2}$ and $ \al_{mn}^{k}=\frac{\lambda}{4\pi  r_{mn}^{k}}$ denote the distance between the $m$-th antenna of user $k$ and the $n$-th  element of the outer surface of the SIM and the free space path-loss coefficient, respectively.   Notably, $r_{mn}^{k}$ can be written as
 \begin{align}
 r_{mn}^{k}=\|[\hat{d}_{k}\cos \hat{\vartheta}_{k}, \hat{d}_{k}\sin\hat{\vartheta}_{k}+(m-1)\bar{d_{m}},0]^{T}-	\bs_{n}\|_{2},
 \end{align}
 where $\hat{d}_{k}$ and $\hat{\vartheta}_{k}=\arctan(y_{k}/x_{k})$ are the distance and angle of user $k$, respectively.
 \begin{remark}
 	 The term $r_{mn}^{k}$ includes both the distance and angle information of user $k$. This property results in performance and coverage improvements in multi-user wireless systems since the SIM can transmit the signal not only at a certain angle but also at a certain distance.
 	 \end{remark}
 \begin{remark}
If we assumed parallel waves, i.e., a far-field LoS channel between the SIM and user $k$, it would be written as
\begin{align}
\bH_{k}^{\mathrm{far}}=\sqrt{\beta_{k}MN}\bee_{k}(\al_{k})\bee_{\mathrm{SIM}}^{\H}(\varphi_{k},\vartheta_{k}),\label{farfield}
\end{align}
where $\beta_{k}$ is the path-loss coefficient between the SIM and user $k$, which is the same as in the case of near field. Also, $\varphi_{k}$, $\vartheta_{k}$ denote the azimuth AoD and the elevation AoD with respect to the SIM, and $\al_{k}$  is the AoA concerning  user $k$. The AoA and AoD are defined based on the location of user $k$ with respect to the SIM. Moreover, by denoting $\textcolor{black}{\kappa}=2 \pi/\lambda$, we have
\begin{align}
\bee_{k}(\al_{k})&\!=\![1, e^{j \textcolor{black}{\kappa} d_{m}\cos\al_{k}}, \cdots, e^{j \textcolor{black}{\kappa} (M-1) d_{m}\cos\al_{k}}]^{\T}\\
\bee_{\mathrm{SIM}}(\varphi_{k},\vartheta_{k})&\!=\![1,\! e^{j \textcolor{black}{\kappa} d_{S} \sin \varphi_{k} \sin \vartheta_{k} }, \cdots\!, e^{j \textcolor{black}{\kappa} (N_{x}\!-\!1) d_{S}\sin \varphi_{k} \sin \vartheta_{k} }]^{\T}\nn\\
	&\!\otimes\! [1, e^{j \textcolor{black}{\kappa} d_{S}\cos \vartheta_{k}}, \cdots, e^{j \textcolor{black}{\kappa} (N_{y}-1) d_{S}\cos \vartheta_{k}}]^{\T}.
	\end{align}
 \end{remark}
 \begin{remark}
 	The near-field LoS channel allows an increase in DoFs by providing multiple information streams since the rank of the LoS channel in the near-field case can be approximated with the DoFs of the  spheroidal waves, which can also be enhanced with distance \cite{Miller2000}. On the contrary, the far-field LoS MIMO channel in \eqref{farfield} is just of rank one.
 \end{remark}
 \section{Downlink Data Transmission and Achievable Rate} 
 In this section, we  present the downlink transmission model, which will lead to the corresponding  achievable rate of a SIM-assisted near-field MIMO system.
 
 \subsection{Downlink Data Transmission}
 During the downlink transmission, wave-based beamforming takes place thanks to the SIM, which is in contrast to conventional digital precoding, where  each symbol is assigned to an individual beamforming vector. Under this setting, the BS selects a set of $K$ appropriate antennas from the total of $M_{\mathrm{BS}}$ antennas since  each data stream has to be transmitted directly from the corresponding antenna at the BS. For the sake of simplicity, we assume that $M_{\mathrm{BS}}=KM$ \cite{An2023b}.
 
 The received signal by user $k$ is given by
 	\begin{align}
 	\by_k&=\bH_{k}\bG\bW^{1}_{k}\bP_{k}\bx_{k}+\bH_{k}\bG\sum_{i\ne k}^{K}\bW^{1}_{i}{\bP_{i}}\bx_{i}+\bn_{k},
 \end{align}
 where $\bn_{k} \in \mathbb{C}^{M\times 1}$ is AWGN  with distribution $ \mathcal{CN}(\b0, \sigma^{2}\Id_{M})$ and  $\bx_{k} \in \mathbb{C}^{M\times 1}$  is the symbol vector with  $\EE\{\bx_{k}\bx_{k}^{\H}\}=\Id_{M}$ and $\EE\{\bx_{k}\bx_{k'}^{\H}\}=\b0$, for $k\ne k'$. Also, $\bP_{k}\in \mathbb{C}^{M\times M}$ is a diagonal matrix, where its $m$-th diagonal entry  denotes the square root of the power allocated to the $m$-th data stream of the $k$-th user. The total transmit power constraint at the BS is given by
 \begin{align}
 	\sum_{k=1}^{K}\|\bP_{k}\|_{F}^{2}\le P,
 \end{align}
 where $P$ denotes the transmit power budget at the BS.
 
 \subsection{Achievable Rate} 
 The achievable rate of $k$-th user is given by
 \begin{align}
 	R_{k}(\bP, \bPhi)=\log |\Id+\bH_{k}\bW^{1}_{k}\bP^{2}_{k}(\bW^{1}_{k})^{\H}\bH_{k}^{\H}\bQ_{k}^{-1}|,
 \end{align}
 where $ \bP=\{ \bP_{i}, \forall i \in \mathcal{K}\}$, $ \bPhi=\{ \bPhi^{l}, \forall l \in \mathcal{L}\}$, and $\bQ_{k}$ is the interference-plus-noise covariance matrix given by
 \begin{align}
 	\bQ_{k}=\sum_{i\ne k}^{K}\bH_{k}\bG\bW^{1}_{i}\bP^{2}_{i}(\bW^{1}_{i})^{\H}\bG^{\H}\bH_{k}^{\H}+\sigma^{2}\Id_{M}.
 \end{align}

	\section{Problem Formulation  and  Optimization}\label{PSConfig}
	Herein, we first formulate the optimization problem, and then we solve it by using the BCD method.
		\subsection{Problem Formulation}
	Herein, we consider the  maximization of the weighted sum rate  by optimizing the transmit power and the phase-shifts of the SIM. In particular, the optimization problem can be formulated as
		\begin{subequations}\label{eq:subeqns}
				\begin{align}
						(\mathcal{P}1)~~&	\max_{\bP,\bPhi }\sum_{i=1}^{K} \eta_{i}	R_{i}(\bP, \bPhi)\label{Maximization1} \\
					&~	\mathrm{s.t}~~~	\bG=\bPhi^{L}\bW^{L}\cdots\bPhi^{2}\bW^{2}\bPhi^{1}\bW^{1},
					\label{Maximization3} \\
										&\;\quad\;\;\;\;\;\!\!~\!		\bPhi^{l}=\diag(\phi^{l}_{1}, \dots, \phi^{l}_{N}), l \in \mathcal{L},
					\label{Maximization5} \\
									&\;\quad\;\;\;\;\;\!\!~\!		|	\phi^{l}_{n}|=1, n \in \mathcal{N}, l \in \mathcal{L},	\label{Maximization7} \\
					&	\;\quad\;\;\;\;\;\!\!~\!			\sum_{k=1}^{K}\|\bP_{k}\|_{F}^{2}\le P	\label{Maximization8},
												\end{align}
	\end{subequations}
	where the constraints \eqref{Maximization3}-\eqref{Maximization7} concern the SIM elements, while the constraint \eqref{Maximization8} corresponds to the power constraint of the BS. Note that $\eta_{k}$ denotes the access priority for user $k$.
	
	 The coupling between  the optimization variables of the transmit power and phase shifts of each surface, which also have constant modulus constraints,  and the non-convexity  of problem $(\mathcal{P}1)$,  make the solution challenging. To this end, we reformulate $(\mathcal{P}1)$ in terms of the mean-square error (MSE).
	 
	 Let $\bU_{k}\in \mathbb{C}^{M \times M}$ be the linear combining matrix for  user $k$, the estimated signal vector of user $k$ can be written as
	 \begin{align}
	\tilde{ \bx}_{k}=\bU_{k}^{\H} \br_k.
	 \end{align}
	 
	 Application of the  weighted minimum mean square error (WMMSE) technique allows to transform \eqref{Maximization1} into a more tractable expression. Thus, $(\mathcal{P}1)  $ can be written as \cite{Shi2011}
		\begin{subequations}\label{eq:subeqns1}
		\begin{align}
			(\mathcal{P}2)~~&	\min_{\bZ,\bU, \bP,\bPhi }\sum_{i=1}^{K} \eta_{i}	g_{i}(\bZ,\bU \bP,\bPhi)\label{Maximization10} \\
			&~	\mathrm{s.t}~~~	\eqref{Maximization3}-\eqref{Maximization8},
				\end{align}
	\end{subequations}
	 where the function $	g_{i}(\bZ,\bU, \bP,\bPhi)$ is given by
	 \begin{align}
	 		g_{i}(\bZ,\bU, \bP,\bPhi)=\log|\bZ_{i}|-\tr(\textcolor{black}{\bZ_{i}} \bE_{i})+M\label{Gfunction}
	 \end{align}
	 with 	 $\bU=\{\bU_{k}, \forall i \in \mathcal{K}\}$,  $\bZ=\{\bZ_{i}\succeq \b0, \forall i \in \mathcal{K}\}$, and $ \bE_{k}$ corresponding to the  set of combining matrices, the set of auxiliary matrices, and the MSE of user $k$.

	 The MSE  of  user $k$  is obtained as
	 \begin{align}
	 \bE_{k}&=\EE\{(\tilde{ \bx}_{k}-\bx_{k})(\tilde{ \bx}_{k}-\bx_{k})^{\H}\}\nn\\
	 &=(\bU_{k}^{\H}\bH_{k}\bG\bW^{1}_{k}\bP_{k}-\Id_{M})(\bU_{k}^{\H}\bH_{k}\bG\bW^{1}_{k}\bP_{k}-\Id_{M})^{\H}\nn\\
	 &+\bU_{k}^{\H}\bH_{k}\bG\bW^{1}_{\bar{k}}\bP^{2}_{\bar{k}}(\bW^{1}_{\bar{k}})^{\H}\bG^{\H}\bH_{k}^{\H}\bU_{k}+\sigma^{2}\bU_{k}^{\H}\bU_{k},\label{MSEmatrix}
	 \end{align}
	 where $\bP_{\bar{k}}=\diag(\bP_{1}, \ldots,\bP_{k-1}, \bP_{k+1}, \ldots, \bP_{K}) $ and $\bW^{1}_{\bar{k}}=[\bW^{1}_{1}, \dots,\bW^{1}_{k-1}, \bW^{1}_{k+1}, \ldots, \bW^{1}_{K}] \in \mathbb{C}^{N\times (K-1)M}$.
	 
	 The tractability of the reformulated optimization problem $	(\mathcal{P}2)$ because of the concavity of the objective function with respect to $\bZ,\bU, \bP $, $\bPhi$ allows its solution by using the BCD method in terms of four blocks, where in each iteration, one variable is optimized while keeping the other variables fixed.
	 
	 \subsubsection{Optimization with respect to $\bU$} The optimal $\bU$ is obtained from \eqref{Maximization10} by solving $\partial 	g_{i}/ \partial \bU_{i}=0$,  $\forall i \in \mathcal{K} $ while keeping $\bZ, \bP ,\bPhi$ fixed. Specifically, we obtain
	 \begin{align}
	 	\bU^{\mathrm{opt}}_{k}=(\bQ_{k}+\bH_{k}\bW^{1}_{k}\bP^{2}_{k}(\bW^{1}_{k})^{\H}\bH_{k}^{\H})^{-1}\bH_{k}\bW^{1}_{k}\bP_{k}.\label{Uopt}
	 \end{align}
	 \subsubsection{Optimization with respect to $\bZ$} In a similar way,  the optimal $\bZ$ is obtained from \eqref{Maximization10} by solving $\partial 	g_{i}/ \partial \bZ_{i}=0$,  $\forall i \in \mathcal{K} $ while keeping $\bU, \bP ,\bPhi$ fixed. Thus, we obtain
	 \begin{align}
	 	\bZ^{\mathrm{opt}}_{k}=	(\bE^{\mathrm{opt}}_{k})^{-1},\label{zopt}
	 \end{align}
	 where $\bE^{\mathrm{opt}}_{k}$ results by substituting $\bU^{\mathrm{opt}}_{k}$ into \eqref{MSEmatrix}.

	 Note that by using the Woodbury matrix identity, the objective function $	(\mathcal{P}2)$ can be rewritten as 
	 \begin{align}
	 &g_{i}(\bZ,\bU, \bP,\bPhi)=\log \det((\bE^{\mathrm{opt}}_{k})^{-1})\nn\\
	 &=\log \det(\Id_{M}+\bH_{k}\bW^{1}_{k}\bP^{2}_{k}(\bW^{1}_{k})^{\H}\bH_{k}^{\H}\bQ_{k}^{-1}),\label{Wood}
	 \end{align}
	 where we have substituted \eqref{Uopt} and \eqref{zopt}. We observe that \eqref{Wood} equals $R_{i}(\bP, \bPhi)$, which means that $	(\mathcal{P}1)$ and $	(\mathcal{P}2)$ are equivalent regarding the  solutions of $\bP$ and $\bPhi $.
	 \subsubsection{Optimization with respect to $\bP$} In this case, with $\bZ, \bU, \bPhi$  assumed fixed, the weighted MSE minimization problem   $(\mathcal{P}2)$ can be written as
	 	\begin{subequations}\label{eq:subeqns2}
	 	\begin{align}
	 		(\mathcal{P}2.1)~~&	\min_{\bP }\sum_{i=1}^{K}	g(\bP)\label{Maximization11} \\
	 		&~	\mathrm{s.t}~~~	\eqref{Maximization8},
	 	\end{align}
	 	\end{subequations}
	 	where $	g(\bP)=\sum_{i=1}^{K}(\tr(\bP_{i}^{\H}\bA \bP_{i})$ $-2\tr(\Re(\bB_{i}\bP_{i})))$, in which, we have denoted $\bB_{i}=\eta_{i}\bZ_{i}\bU_{i}^{\H}\bH_{i}\bG\bW^{1}_{i}$ and $\bA=\sum_{k=1}^{K} \eta_{k}(\bW^{1}_{k})^{\H} \bG^{\H}\bH_{k}^{\H}\bU_{k}\bZ_{k} \bU_{k}^{\H}\bH_{k}\bG\bW^{1}_{k} $. 
	 	
	 	The problem 	$(\mathcal{P}2.1)$ is a a quadratically constrained quadratic program, which can be solved by using CVX \cite{Grant2015}. Also, 	$(\mathcal{P}2.1)$ could be solved in closed form with reduced complexity by applying the Lagrangian dual decomposition method \cite{Pan2020}.
	 	
\subsubsection{Optimization with respect to $\Phi$} Given $\bZ, \bU, \bP$, after substituting \eqref{MSEmatrix} into \eqref{Gfunction}, the optimization problem with respect to the phase shifts matrices is formulated as
\begin{subequations}\label{eq:subeqns3}
	\begin{align}
		(\mathcal{P}2.2)~~&	\min_{\bphi_{l}}\sum_{i=1}^{K}	g_{i}(\bphi_{l})\label{Maximization112} \\
		&~	\mathrm{s.t}~~~	\eqref{Maximization3}-\eqref{Maximization7},
	\end{align}
\end{subequations}
where 
\begin{align}
	g_{i}(\bphi_{l})=\tr(\bG^{\H}\bC_{i} \bG \bD -2\Re(\bE_{i}\bG)).\label{optPhi}
\end{align}
	with  $\bC_{i}=\eta_{i}\bH_{k}^{\H}\bU_{k}\bZ_{k} \bU_{k}^{\H}\bH_{k}$, $\bD=\bW^{1}_{i}(\sum_{k=1}^{K}\bP_{k}^{2})(\bW^{1}_{i})^{\H}$, and  $\bE_{i}=\eta_{i}\bW^{1}_{i}\bP_{i}\bZ_{k} \bU_{k}^{\H}\bH_{k}$
	 The problem $	(\mathcal{P}2.2) $ is non-convex due to the unit-modulus constraint and the non-convex constraints regarding the phase shifts.
	 
	 Herein, we are going to  apply 	  the projected gradient 	ascent algorithm until convergence to a  locally optimal solution to $ (\mathcal{P}2.2)$. Specifically, by starting from $ \bphi_{l}^{0} $, we shift along the gradient of   $  	g_{i}(\bphi_{l})$. Next, we project the new point $ \bphi_{l} $ onto  $ \Phi_{l} $ to hold the new points in the feasible set. For ease of exposition, we have defined the set $\Phi_{l}=\{|	\phi^{l}_{n}|=1, n \in \mathcal{N}\} $. Note that $ \phi^{l}_{n} $ has to be found inside the unit circle  because of the unit-modulus constraint. Also, we denote  $ P_{\Phi_{l}}(\cdot) $ the projection onto $ \Phi_{l} $. 
	 
	 	The following iteration describes the algorithm. In particular, we have
	 \begin{align}
	 	\bphi_{l}^{i+1}&=P_{\Phi_{l}}(\bphi_{l}^{i}+\mu_{i}\nabla_{\bphi_{l}}g_{j}(\bphi_{l}^{i}))\label{p1}.
	 \end{align}
	 The step size is obtained by the Armijo-Goldstein backtracking line search method. We have $ \mu_{i} = T_{i}\bar{\kappa}^{m_{i}} $, where   $ \bar{\kappa} \in (0,1) $ and $ T_i>0 $ with $ m_{i} $ being  the
	 smallest positive integer that  satisfies
	 \begin{align}
	 	g_{j}(\bphi_{l}^{i+1})\geq	B_{T_{i}\bar{\kappa}^{m_{i}}}(\bphi_{l}^{i};\bphi_{l}^{i+1}),
	 \end{align}
	 where 
	 \begin{align}
	 	\!\!	B_{\mu}(\bphi_{l};\bx)\!=\!g_{j}(\bphi_{l})\!+\!\langle	\nabla_{\bphi_{l}}g_{j}(\bphi_{l}),\bx\!-\!\bphi_{l}\rangle\!-\!\frac{1}{\mu}\|\bx-\bphi_{l}\|^{2}_{2}
	 \end{align}
	 is the  quadratic approximation of $ f(\bphi_{l}) $.
	 
	 The following lemma provides the  complex-valued gradient.
	 \begin{lemma}\label{LemmaGradient}
	 	The complex gradient $ \nabla_{\bphi_{l}}	g_{i}(\bphi_{l})$ is given in closed-form by
	 		 		\begin{align}
	 			\nabla_{\bphi_{l}}	g_{i}(\bphi_{l})=\diag(	   \bJ_{l}^{*} \bD^{\T} \bG^{\T}  \bC_{i}^{\T}  \bR_{l}^{*} 	  ).
	 		\end{align}
	 		\end{lemma}
	\begin{proof}
	To obtain $ \nabla_{\bphi_{l}}	g_{i}(\bphi_{l})$, we first derive the following differential
	\begin{align}
	d(	g_{j}(\bphi_{l}^{i}))&=\tr\left(\right.\!\!d(\bG^{\H})\bC_{i} \bG \bD +\bG^{\H}\bC_{i} d(\bG) \bD\nn\\
	&-2\Re(\bE_{i}d(\bG))\!\!\left.\right)\\
	&=\tr\left(\right.\!\!\bJ_{l}^{\H}d(\bPhi_{l}^{\H})	\bR_{l}^{\H}
	\bC_{i} \bG \bD +\bG^{\H}\bC_{i} \bR_{l} d(\bPhi_{l})\bJ_{l} \bD\nn\\
	&-2\Re(\bE_{i}\bR_{l} d(\bPhi_{l})\bJ_{l})\!\!\left.\right)\label{num2},
	\end{align}
	where, in  \eqref{num2}, we have substituted $d(\bG)=\bR_{l} d(\bPhi_{l})\bJ_{l}$ with $ 	\bR_{l}=\bPhi_{L}\bW^{L}\cdots\bPhi_{l+1}\bW^{l+1} $, and $\bJ_{l}= \bW^{l}\bPhi_{l-1}\bW^{l-1}\cdots \bPhi_{1} $.
	
	Having derived the differential, we obtain
	\begin{align}
		\nabla_{\bphi_{l}}g_{i}(\bphi_{l})&=\frac{\partial}{\partial\bphi_{l}^{*}}g_{i}(\bphi_{l})=\diag(	   \bJ_{l}^{*} \bD^{\T} \bG^{\T}  \bC_{i}^{\T}  \bR_{l}^{*} 	  ).
	\end{align}
\end{proof}
	 
\subsubsection{Overall Algorithm} 
	\begin{algorithm}[th]
	\caption{Overall Algorithm\label{Algoa1}}
	\begin{algorithmic}[1]
		\STATE Initialize feasible $\bP$ and $\bPhi$ that satisfy \eqref{Maximization3}-\eqref{Maximization8},
		\REPEAT
		\STATE Given $\bP$ and $\bPhi$, update the combining matrices $\bU$ using \eqref{Uopt}.
				\STATE Given $\bP$, $\bPhi$, and $\bU$ update the auxiliary matrices $\bZ$ using \eqref{zopt}.
		\STATE Given  $\bPhi$,  $\bU$, and $\bZ$ update the transmit power $\bP$ by solving \eqref{eq:subeqns2}.
				\STATE Given  $\bP$,  $\bU$, and $\bZ$ update the phase shifts  matrices $\bPhi$ by solving \eqref{eq:subeqns3}.
				\UNTIL{The increase of \eqref{eq:subeqns1} is less than a threshold $\epsilon$.}
			\end{algorithmic}
\end{algorithm} 
A summary of the  proposed algorithm for solving \eqref{eq:subeqns1} is provided in Algorithm \ref{Algoa1}.
 Given that \eqref{Maximization10} is non-decreasing in each iteration  and since it  is upper bounded because the transmit power is limited,  Algorithm \ref{Algoa1} is guaranteed to converge to a stationary point. 

The complexity of Algorithm \ref{Algoa1} is mainly yielded from Problems $	(\mathcal{P}2.1)$ and $	(\mathcal{P}2.2)$. In particular,  the complexity of  $	(\mathcal{P}2.1)$ is $O(K M_{\mathrm{BS}}^{3} )$ and the complexity  $	(\mathcal{P}2.2)$ is $O(KL(N^2+M^2))$. Hence, the overall computational complexity of Algorithm \ref{Algoa1} is provided by $O(I_{\mathrm{BCD}} K (M_{\mathrm{BS}}^{3} +L(N^2+M^2))$, where $I_{\mathrm{BCD}}$ is the number of BCD iterations.	 
	
	\section{Numerical Results}\label{Numerical}
Herein, we assess the performance of the weighted sum rate of a  SIM-assisted near-field multi-user MIMO communication system by depicting analytical results and Monte Carlo simulations. The layout assumes that the SIM is parallel to the $ y-z $ plane and centered  along the $ x-$axis. The users are located in a circular region with a radius between $2$ m and $4$ m. The spacing  between adjacent meta-atoms  is assumed to be $ \lambda/2 $  while the size of each meta-atom  is $ \lambda/2 \times \lambda /2 $. Also, we have $ d_{\mathrm{SIM}}= T_{\mathrm{SIM}}/L$, where $ T_{\mathrm{SIM}}=5 \lambda $ is the thickness of the SIM. Regarding the  distance between the $ \tilde{n}-$th meta-atom of the $ (l-1)-$st metasurface and the  $ {n}-$th meta-atom of the $ l-$st metasurface, it  is given by $ d_{n,\tilde{n}}^{l}=\sqrt{d_{\mathrm{SIM}}^{2}+d_{n,\tilde{n}}^{2}} $, where
\begin{align}
	\!\!	d_{n,\tilde{n}}\!=\!\frac{\lambda}{2}\sqrt{\lfloor | n-\tilde{n}|/N_{x}\rfloor^{2}\!+\![ \mathrm{mod}(|n-\tilde{n}|,N_{x})]^{2}}.
\end{align}
In the case of the  transmission distance between the $ m $-th antenna of the BS and the $ \tilde{n} $-th meta-atom on the first metasurface layer, it  is given by \eqref{dist1}. Moreover, we have $ \cos x_{n,\tilde{n}}^{l}= d_{\mathrm{SIM}}/ d_{n,\tilde{n}}^{l}, \forall l \in \mathcal{L}$.
\begin{figure*}
	\begin{align}
		{\small 	d_{\tilde{n},m}^{1}\!=\!\sqrt{\!d_{\mathrm{SIM}}^{2}\!+\!\Big[\!\Big(\!\mathrm{mod}(\tilde{n}\!-\!1, N_{x})\!-\!\frac{N_{x}\!-\!1}{2}\!\Big)\frac{\lambda}{2}\!-\!\Big(m\!-\!\frac{N_{t}\!+\!1}{2}\Big)\frac{\lambda}{2}\Big]^{2}\!+\!\Big(\!\lceil \tilde{n}/N_{x} \rceil\!-\!\frac{N_{y}\!+\!1}{2}\Big)^{2}\frac{\lambda_{2}}{4}}}.\label{dist1}
	\end{align}
	\hrulefill
\end{figure*}
The  path loss   is given by $	\tilde \beta_k = C_{0} (d_k/\hat{d})^{-\alpha}$, 
where $ C_{0} =(\lambda_{2}/4 \pi \hat{d})$ is  the free space path loss at the reference distance of $ \hat{d}=1~ \mathrm{m}$, and $\alpha=2.5$ is the path-loss exponent.   The  carrier frequency and the system bandwidth are $ 10~\mathrm{GHz} $ and $ 20~\mathrm{MHz} $, respectively.  Furthermore, we assume $  M_{\mathrm{BS}} =4 $, $ K=4 $, $M=2$, $ N=40 $, and  $ L=4 $. The Rayleigh distance of a SIM with $N=40$ elements at the above frequency of $5$ GHz is approximately $5$ m, which means that all users are inside the region of the near field.  

To show better the improvement in the near field, we assume two settings: i) a random user setup, where users lie in different angles, and ii) an inline user setup, where users lie in the same line, i.e., they are located in the same angle with respect to the SIM.

	\begin{figure}%
	\centering
	\includegraphics[width=0.9\linewidth]{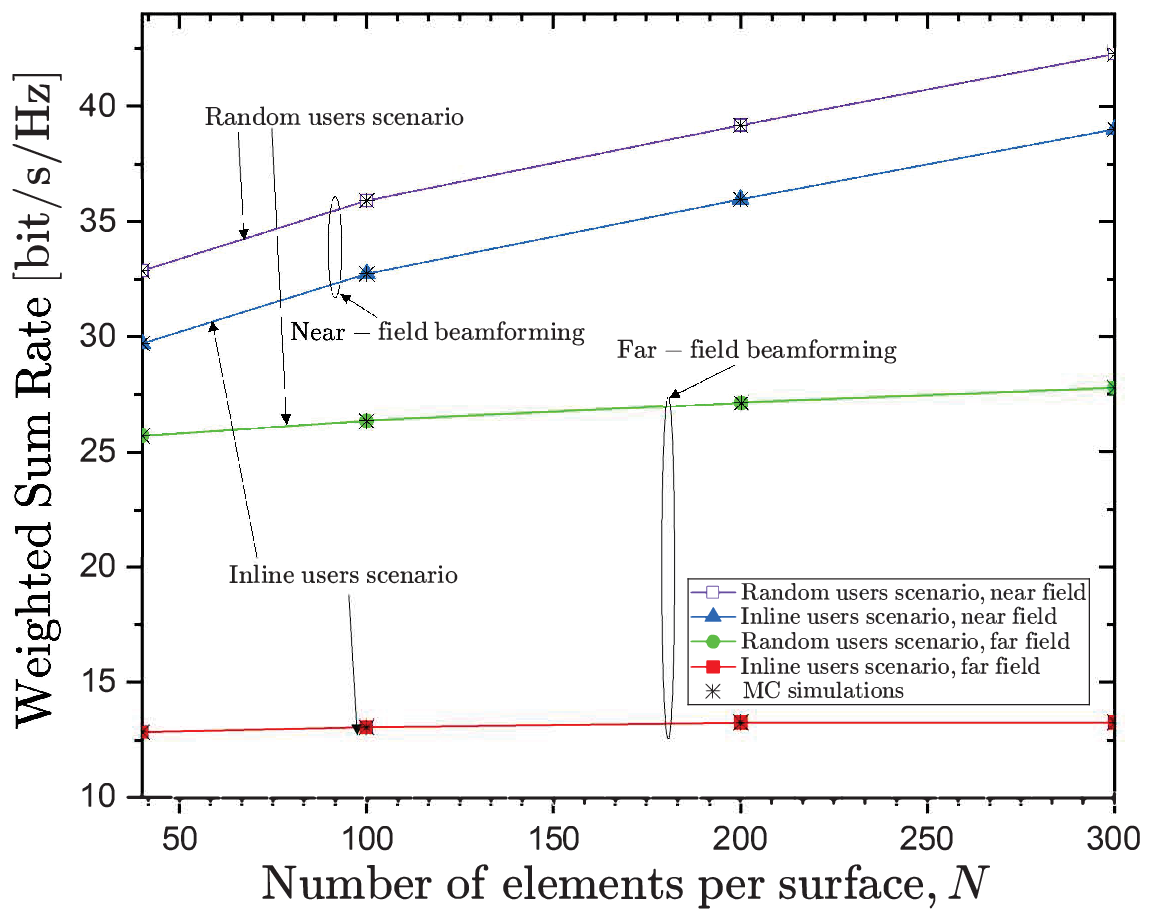}
	\caption{Weighted sum rate of the  SIM-aided MIMO architecture with respect to the number of meta-atoms  $ N $.}
	\label{fig2}
\end{figure}
In Fig. \ref{fig2}, we depict the achievable weighted sum rate  versus the nu6mber of   elements $ N $ of each surface while varying the number of surfaces $L$. For the sake of comparison, we have also depicted the far-field scenario obtained based on \eqref{farfield}. In this case, the angles of the channels  are obtained based on the Uniform distribution. We observe that the rate increases with $N$ because a larger $N$ results in a higher beamforming gain but in the case of the far-field, especially in the inline user setting, the  increase is insignificant because both the gain and the interference become higher with $N$. Moreover, the near-field beamforming achieves a higher rate than the far-field one because the near-field LoS channels transmit more data streams to the users since they have a  higher rank. In other words, a degrees of freedom enhancement appears.  Also, under the same type of beamforming, i.e., in the near-field or far-field beamforming, the random user scenario appears better performance compared to the inline user scenario because the latter comes with higher inter-user interference. The reason is that in the inline user scenario, only the user distance information is leveraged to mitigate the inter-user interference, while in the random user scenario both the distance and angles information are leveraged. Notably, Monte Carlo (MC) simulations verify the analytical results.

In Fig. \ref{fig3}, we show  the achievable weighed sum rate versus the transmit power while varying the number of surfaces $L$ in the scenario of random users setup. As can be seen, the weighted sum rate increases in all cases with power because a higher power budget enables the reception of stronger signals from the users. Also, the SIM outperforms significantly the case of a single surface, i.e., $ L=1$. Generally, by increasing the number of surfaces of the SIM, the rate increases because the SIM can manage to mitigate the inter-user interference in the wave domain \cite{An2023b}. Furthermore, we show the comparison between near and far-field beamforming, and we observe that the former performs better because the latter includes higher inter-user interference.

\begin{figure}%
	\centering
	\includegraphics[width=0.9\linewidth]{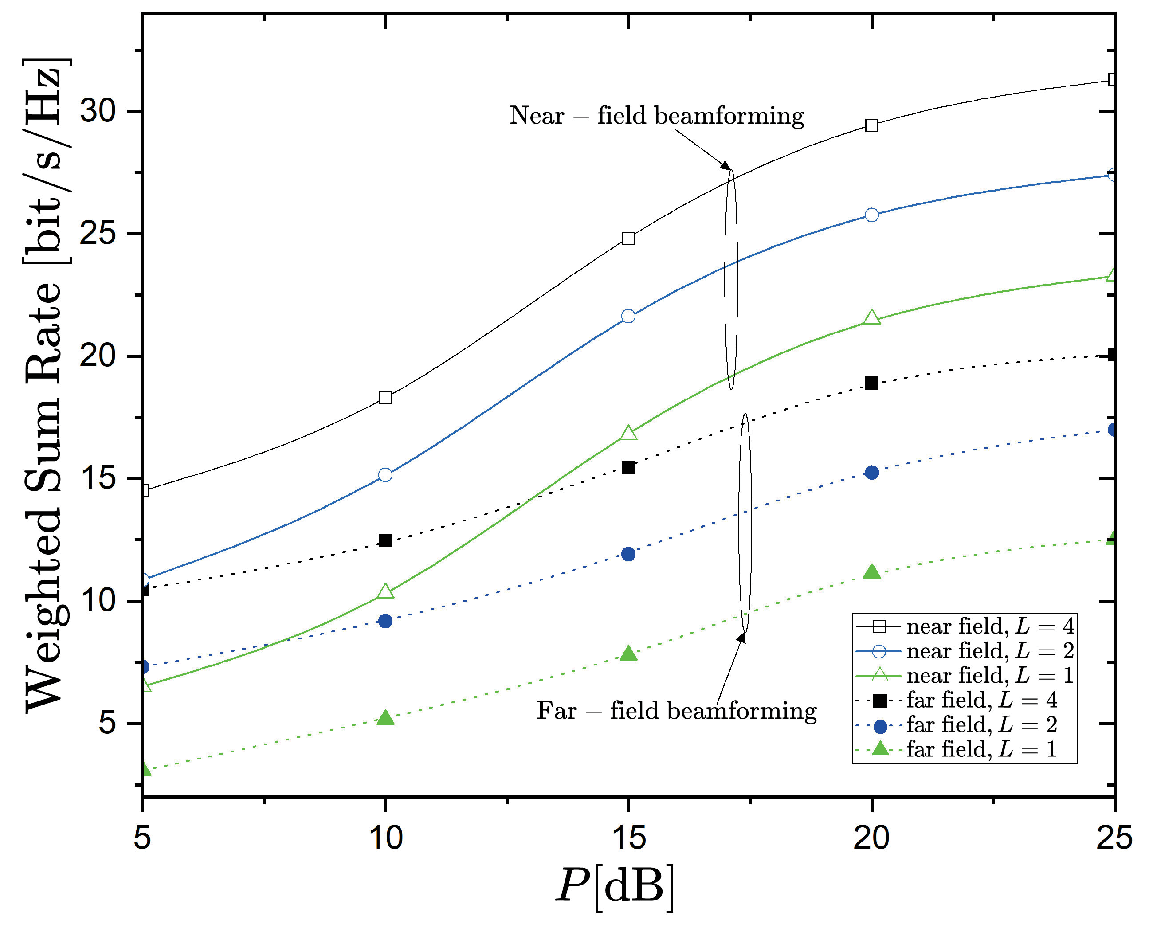}
	\caption{Weighted sum rate of the  SIM-aided MIMO architecture with respect to the transmit power  $ P $.}
	\label{fig3}
\end{figure}
 
	\section{Conclusion} \label{Conclusion} 
This paper presented the study of the weighted sum rate of SIM-assisted multiuser MIMO communication systems in the near field. A BCD-based algorithm was used to solve the non-convex optimization. Numerical results showed that near-field beamforming can improve the performance compared to far field and  the degrees of freedom are facilitated by the near-field channels.
	 
	\bibliographystyle{IEEEtran}
	
	\bibliography{IEEEabrv,bibl}
	\end{document}